\documentclass[conference,a4paper]{IEEEtran}
\usepackage[utf8]{inputenc} 
\usepackage[T1]{fontenc}
\usepackage{url}
\usepackage{ifthen}
\usepackage{cite}
\usepackage[cmex10]{amsmath} 
\usepackage{amsmath,amssymb,amsthm}
\usepackage{layout}
\usepackage{verbatim}
\usepackage{amsfonts}
\usepackage{graphicx}
\usepackage{graphicx}

\newtheorem{remark}{Remark}
\newtheorem{assumption}{Assumption}
\newtheorem{definition}{Definition}
\newtheorem{proposition}{Proposition}

\DeclareMathOperator{\SINR}{SINR}
\DeclareMathOperator*{\LimSup}{limsup}
\date{}
\begin{document}

\title{Queue-aware Energy Efficient Control for Dense Wireless Networks}

  \author{%
  \IEEEauthorblockN{Maialen Larra\~naga, Mohamad Assaad and Koen De Turck}
  \IEEEauthorblockA{Laboratoire des Signaux et Syst\`emes (L2S,CNRS), CentraleSup\'elec\\
                    Gif-sur-Yvette, France}
}

\maketitle

\begin{abstract} We consider the problem of long term power allocation in  dense wireless networks. The framework considered in this paper is of interest for machine-type communications (MTC). In order to guarantee an optimal
	operation of the system while being as power efficient as possible, the
	allocation policy must take into account both the channel and queue states of the devices. This is a complex stochastic optimization problem, that can be cast
	as a Markov Decision Process (MDP) over a huge state space. In order to tackle
	this state space explosion, we perform a mean-field approximation on the MDP.
	Letting the number of devices grow to infinity the MDP converges to a
	deterministic control problem. By solving the Hamilton-Jacobi-Bellman Equation, we obtain a well-performing power allocation policy for the original stochastic problem, which turns out to be a threshold-based policy and can then be easily implemented in practice.  
\end{abstract}

\section{Introduction}\label{s1}
The steep increase of the number of mobile devices in use has brought a lot of attention to the design of large wireless networks. The proliferation of Internet of Things (IoT) applications will lead to a drastic increase of the density of devices in future wireless networks. Machine Type Communications (MTC) is the cellular technology for IoT. In 5G (and beyond) networks, it is foreseen that the density of MTC devices may surpass 1 Million of devices per $Km^2$ \cite{5G}. In such dense networks,  the network designer has to deal with severe interference issues in order to guarantee a certain level of quality of service (QoS). This can be handled by advanced physical layer solutions (e.g. Interference Alignment \cite{IA2008}, etc.),  which may however suffer in some cases from high complexity or high signaling overhead. Furthermore, opportunistic resource allocation, such as power control, can also help manage the impact of interference among users and hence improve their QoS. The focus of this paper is on resource allocation in such dense networks. The problem of power control in wireless networks has been widely studied in the past, e.g. in \cite{chiang} and the references therein. The problem of power control in large scale networks has also been investigated in the past using game theory and mean-field games, e.g. \cite{mf1,mf2,mf3}. The problem in these references is first formulated as a stochastic differential game and then the sufficient conditions for the existence and uniqueness of the mean-field equilibrium are provided. It is also shown that this equilibrium power can be obtained by solving  a coupled system of Fokker-Planck-Kolmogorov (FPK) equations (which take the form of forward equations) and Hamilton-Jacobi-Bellman (HJB) equations (which take the form of backward equations) to form a system of so-called forward-backward equations. In the aforementioned work on mean-field games, two issues are not addressed: i) solving numerically the resulting forward-backward equations has a high complexity, and ii) the focus of the proposed frameworks is on the wireless links,  i.e. channel state information (CSI), without taking into account the traffic patterns and/or the queues of the users. In fact,  since the CSI reveals the instantaneous transmission opportunities at the physical layer and the queue state information (QSI) reveals the urgency of the data flows, a good control policy must take into account both the CSI and the QSI, and the goal in the present paper is to find such control policy. Queue-aware control problems have been widely studied in the literature and several approaches have been used. For example in \cite{Apostolos,Tassiulas}, the allocation policies are based on the MaxWeight rule which allows to stabilize the queues of the users. However, the MaxWeight rule may suffer from high delay and therefore delay-ware control policies for wireless networks have been developed in \cite{Lau1,Lau2}, where it is established that Markov Decision Processes (MDP) constitute the systematic approach used in the development of the delay-aware policies. A survey on delay-aware control policies can be found in \cite{Lau3}.  MDP problems prove to be a difficult problem to solve. Many techniques have been proposed, for instance brute force value iteration or policy iteration \cite{puterman,Lau3} that find the optimal control policy by solving the Bellman equation. However these techniques have a huge complexity (due to the curse of dimensionality) because solving the Bellman equation involves solving a large system of non-linear equations whose size increases exponentially in the number of users. Effort has been done in order to deal with the curse of dimensionality \cite{Lau1} by utilizing the interference filtering property of the CSMA-like MAC protocol. A closed-form approximate solution and the associated error bound have been derived using perturbation analysis. However this assumption on the weak interference seems constraining and not adapted to dense wireless networks where the interference level cannot be small. In this work, to overcome the dimension problem we use the mean field approach. It consist of neglecting the behavior of individual user by only considering the one of the proportion of users in certain state. This allows us to move from a stochastic optimization problem to a continuous-time deterministic one. 
We formulate the bias optimal control problem based on this deterministic approach and we solve it by characterizing a solution of the Hamilton-Jacobi-Bellman equation.
One of the main challenges we face is that the equations are fully coupled, meaning the solution of one is dynamically depending on the solution of the other. In order to handle those challenges we adopt a three steps method to finally obtain the optimal power control. We first characterize the optimal equilibrium point of the dynamic system with respect to the control variable. Then we prove convexity of the cost function (in all possible equilibrium points). Finally, we propose a threshold type of policy that satisfies the HJB equations, and is hence bias-optimal.  The obtained policy, being a simple threshold type of policy, can easily be applied in the original stochastic system and provides nearly-optimal performance. 

Summarizing, these are the main differences between the present paper and the existing body of literature. While the existing work on mean-field games in wireless networks focuses on the CSI and formulate the power control problems using game theory (e.g. differential game) \cite{mf1,mf2,mf3}, we consider the impact of the QSI in addition to the CSI in this work. Furthermore, our problem is a multi-dimensional stochastic optimization problem that is formulated as an infinite horizon average cost MDP. Last but not least, while most of the existing work on mean field (e.g. \cite{mf1,mf2,mf3}) does not provide a simple solution of the forward-backward equation (resulting from the mean field game) which is known to be complex, we analyze in this paper  the  forward-backward equation resulting from our MDP problem and provide a full characterization of the mean field solution under a specific channel model. This is the main contribution in this paper. Moreover, it is worth mentioning that our obtained policy is a threshold based policy and hence it is easy to implement in practice.

\section{System Model}\label{s2}
In this section, we introduce the system model of our wireless network consisting of $N$ transmitters communicating with a Base Station (BS). The transmitters correspond for example to users or to Machine Type Communication (MTC) devices. We will use the terms \emph{transmitter} and \emph{user} interchangeably throughout the paper. We assume time to be slotted and users to be synchronized to these time slots. At the beginning of each time slot, users that have been allotted enough transmission power will be able to transmit their packets. The latter not only depends on the allocated power but also the channel quality of each user. 
We consider the channel state of transmitter $n$, i.e., $h_{n}(t)$, to take values in the set $\{c_1,\ldots,c_K\}$. We will assume that $h_{n}(t)=c_K$ is the best quality channel and $c_1$ the worst. The channel is further assumed to evolve as an i.i.d.~process from one time-slot to another, although our modeling framework holds for the more general case of Markovian channel dynamics as well. 

The users transmit on the same bandwidth and interfere with each other. For a given channel state $h_{n}(t)$ for user $n$, and transmit power $p_n(t)$ the 
$\SINR$ of user $n$ is given by 
\begin{equation*}
\SINR_{n}(\mathbf{h}(t),\mathbf{p}(t))=\frac{h_{n}(t)p_n(t)}{\sum_{k\neq n}\alpha_k h_{k}(t)p_k(t)+N_0},
\end{equation*}
where $N_0$ is Gaussian noise, $\alpha_k$ is a weight that comes for example from the processing gain
at the receiver (this is widely used in the literature, e.g. in \cite{chiang,mf1}  and in \cite{meshkati2006game} for a CDMA system and a Match Filter receiver), $\mathbf{h}(t)=(h_1(t),\ldots,h_N(t))$ and $\mathbf{p}(t)=(p_1(t),\ldots,p_N(t))$. We will assume that the transmit power of each transmitter in each time slot is bounded by $p_{\text{max}}$. Namely,  $0\leq p_n(t)\leq p_{\text{max}}, \text{ for all } n\in\{1,\ldots,N\}$. For ease of notation we define
$
\SINR_{n}(t):=\SINR_n(\mathbf{h}(t),\mathbf{p}(t)).
$
In order to receive correctly the information at the receiver, it is required that
\begin{equation}\label{eq:thetadefinition}
 \SINR_{n}(t)\geq \theta.
\end{equation}
where $\theta$ is a given threshold. For convenience, we also assume that each transmitter can transmit at most one packet per time slot if the SINR constraint in ~\eqref{eq:thetadefinition} is satisfied. The extension to the case of higher rates is straightforward.  Therefore, the achievable data rate of user $n$ is given by
\begin{equation}\label{eqn2}
R_n(\mathbf{h}(t),\mathbf{p}(t))=\mathbf{1}_{\{\SINR_{n}(t)\geq \theta\}},
\end{equation}

Let us now present the bursty data source and the queue dynamic for each user $n$.
Let  $A_n(t)$  be the (random) number of packet arrivals to the transmitter $n$ at the end of time slot $t$. Let $A_n(t)$ be i.i.d. over time slots. We assume that in each time slot there will be at most one packet arrival, i.e., $P(A_n(t)=1)=\rho$ and $P(A_n(t)=0)=1-\rho$ with $\rho>0$. 

Each transmitter has a data queue for the bursty traffic flow towards its associated receiver.
Let $ Q^\phi_n(t)$ be the queue length at transmitter $n$ at the beginning of time slot $t$ under a power allocation policy $\phi$. The queue dynamic is then given by
\begin{equation}\label{eqn3}
Q_n^\phi(t+1)=\max\lbrace Q_n^\phi(t)-R_n(\mathbf{h}(t),\mathbf{p}(t)),0\rbrace+A_n(t).
\end{equation}
For mathematical tractability, we assume that the queue length cannot exceed $Q_{\text{max}}$ and that packets that arrive during the buffer overflow are dropped. Namely, $Q^\phi_n(t) \in\{0,\ldots,Q_{\text{max}}\}$. 
\begin{remark}\label{Rem1}
The dynamics of all $N$ queues are coupled together due to the interference term in the expression of the $\SINR$. The departure of the queue at each transmitter depends on the power actions of all the other transmitters.
\end{remark}

The objective of the present work is to find an optimal power allocation policy $\phi$ taking into account the interferences between users in the system.

\section{Control Problem Formulation}\label{s3}

Let $X_n^\phi(t)=(h_n(t),Q_n^\phi(t))$, be the state of transmitter $n$, namely, the channel condition and the queue length. The transmit power is dynamically adapted to the global system to handle the interference mitigation. In this work we focus on the set of all stationary policies $\Phi$. 
Given a control policy $\phi\in\Phi$, the stochastic process $X_n^\phi(t)$ is a controlled Markov chain with the following transition probabilities
\begin{align*}
&\nu^\phi_n(X_n^\phi(t+1)=(c,q)|Q_n^\phi(t)=q',\mathbf{h}(t),\mathbf{p}(t))\\
&=\mathbb{P}(h_n(t+1)=c|h_n(t)=c')\\
&\quad\cdot\mathbb{P}(Q_n^\phi(t+1)=q|Q_n^\phi(t)=q',\mathbf{h}(t),\mathbf{p}(t)), \text{ for all }n.\nonumber
\end{align*}
Observe that, for the i.i.d. channel model, the probability  $\mathbb{P}(h_n(t+1)=c|h_n(t)=c')$ reduces to $\mathbb{P}(h_n(t+1)=c)$. 
According to the system model in Section II, we note that in $\mathbb{P}(Q_n^\phi(t+1)=q|Q_n^\phi(t)=q',\mathbf{h}(t),\mathbf{p}(t))$ $q$ can only take three values $q\in\{q'-1,q',q'+1\}$. We give explicit expression of all transition probabilities in Appendix~\ref{prob_expressions}.

The objective is to minimize the average power cost together with the queue length. In order to reduce the delay and queue overflow, users with a higher queue length should be prioritized over users with small number of packets to transmit.
The objective of the present work is then to minimize
\begin{align*}
\mathcal{L}^\phi
&=\limsup_{T\rightarrow\infty}\frac{1}{T}\sum_{t=0}^{T-1}\sum_{n=1}^N\mathbb{E}\left[p_n(t)+\lambda Q^\phi_n(t)\right]
\end{align*}
 where $\lambda \geq 0$ is a weight parameter that can be adjusted to in order to find a tradeoff between the power consumption the queue length minimization. 

This problem, due to the complex interrelations between users, is a very complex MDP problem. Well known simple heuristics to solve such MDPs (such as Whittle's index policy)  fail in this problem, due to the interferences between users. In the next section we therefore develop a mean-field approximation.

\subsection{Mean-Field Approach}
We consider each user in the network as an object evolving in a finite state space, the state of user $n$ at time $t$ is denoted as $X_n^\phi(t)$ and equals $(c,q)$, where $c\in\{c_1,\ldots,c_K\}$ and $q\in\{0,\ldots,Q_{\text{max}}\}$. We assume that the users are distinguishable only through their state. This means that the behavior of the system only depends on the proportion of users in every state. Let $M^N(t)$ be the empirical measure of the collection of users, it is a $S$-dimensional vector with the $n$-th component given by
$
\mathbf{M}^N(t)=(\mathbf{M}_1^N(t),\ldots,\mathbf{M}_K^N(t)),$ where
$\mathbf{M}^N_i(t)=(M^N_{i,0}(t),\ldots,M^N_{i,Q_{\text{max}}}(t)),$  and
$M^N_{i,j}(t)=\frac{1}{N}\sum_{n=1}^N\mathbf{1}_{\{X_n^\phi(t)=(c_i,j)\}},
$
for all $i\in\{1,\ldots,K\}$ and all $j\in\{1,\ldots,Q_{\text{max}}\}$. The value of  $M^N_{i,j}(t)$ is to be interpreted as the proportion of transmitter/users in channel state $c_i$ and queue length $j$. We then have that, the set of possible values for $\mathbf{M}^N$ is the set of probability measures on 
$$
\mathcal{S}=\left\{(c,q): c\in\{c_1,\ldots,c_K\}, q\in\{1,\ldots,Q_{\text{max}}\}\right\}.
$$

The mean field approach allows us to move from a stochastic optimal control problem to a deterministic one. The advantage is that we are no longer in an uncertain environment and we can now overcome the curse of dimensionality due the large number of users in the network. The limiting deterministic optimization problem is formulated as follows. Let us denote by $D\mathbf{M}^N(t)$ the expected drift of $\mathbf{M}^N(t)$, that is,
$$
D\mathbf{M}^\phi:=\mathbb{E}(\mathbf{M}^N(t+1)-\mathbf{M}^N(t)|\mathbf{M}^N(t)).
$$
We now aim at obtaining the explicit expression of the expected drift under the policy $\phi$. In order to do so, let us first define $s_i$ to be the state that corresponds to the $i^{\text{th}}$ entry in $\mathbf{M}^N$. Then we define
$\nu_{i,i}^\phi(\mathbf{m})$  to be the probability that a user in state $s_i\in\mathcal{S}$ at time slot $t$, transitions to state $s_j\in\mathcal{S}$ at time slot $t+1$ given that $\mathbf{M}^N(t)=\mathbf{m}$, that is, 
$$
\nu_{i,j}^\phi(\mathbf{m}):= g_i(\mathbf{m})\gamma_{i,j}^1+(1-g_i(\mathbf{m}))\gamma_{i,j}^0,
$$
where $g_i(\mathbf{m})$ is the fraction of users in state $s_i\in\mathcal{S}$ whose $\SINR_n(t)\geq \theta$, with $n$ a user in state $s_i$. The values of $\gamma_{i,j}^a$ for $a=0,1$ represent the transition probabilities from state $s_i$ to state $s_j$, when the $\SINR_n(t)\geq\theta$ for all users $n$ in state $s_i$ if $a=1$ and, when  $\SINR_n(t)<\theta$ for all users $n$ in state $s_i$ if $a=0$. These values depend on whether we assume an i.i.d. channel evolution model or a Markovian one. For both cases the expressions of $\gamma_{i,j}^a$ for $a=0,1$ can be found in Appendix B. 

We then have 
$
D\mathbf{M}^N(t)\bigg|_{\mathbf{M}^N(t)=\mathbf{m}}=\sum_{i}\sum_{j} \nu_{i,j}^\phi(\mathbf{m})\vec e_{ij}=U^\phi(\mathbf{m})\mathbf{m},
$
where $\vec e_{ij}=(0,\ldots,0,\overbrace{-1}^{i^\text{th}},0,\ldots,0,\overbrace{1}^{j\text{th}},0,\ldots)$, i.e., the $K\cdot (Q_{\text{max}}+1)$ dimensional vector with a $-1$ entry in the $i^{\text{th}}$ position and the entry at $j^{\text{th}}$ position equal to $1$ and $I$ is the identity matrix. We further have $\vec e_{ii}=\vec 0$. Also note that  
\begin{align*}
U_{i,j}^\phi(\mathbf{m})=\begin{cases}
-\sum_{r\neq i}\nu_{i,r}^\phi(\mathbf{m}) &\hbox{if } i=j,\\
\nu_{j,i}^\phi(\mathbf{m}) &\hbox{if } i\neq j.
\end{cases}
\end{align*}
We can now define
$
\mathbf{m}(t+1)-\mathbf{m}(t)=U_{i,j}^\phi(\mathbf{m}(t))\mathbf{m}(t).
$
The latter can be seen as a {\it fluid system} which is defined for any $\mathbf{m}(t)$ and not only for probability densities.

In the original stochastic problem we aim at minimizing the long run expected average power and the queue length. Note that in the fluid setting there are several power and queue trajectories that reach to the same equilibrium point and hence we aim at minimizing the {\it biased cost}. 
Assuming that $\phi$ is such that all users in same state $s\in\mathcal S$ are allocated same power and that $\mu_n=\mu_{n'}$ if user $n$ and $n'$ are both in the same state $s\in\mathcal{S}$, we can equivalently write
\begin{align}\label{eq:objective_meanfield}
&\mathcal{L}^{\phi,N}\approx\limsup_{T\to\infty}\frac{N}{T}\sum_{t=0}^{T-1}\sum_{i=1}^{K\cdot Q_{\text{max}+1}}\mathbb{E}(p_i(t)+\lambda m_i(t)\sigma(i)),
\end{align}
where $\sigma(\cdot)$ is a mapping between $i\in\{1,K\cdot(Q_{max}+1)\}$ and the queue-length. Namely, if $i=z*K+j$ then $\sigma(i)=j$ for all $z\in\{0,\ldots, Q_{\text{max}}+1\}$.

For the mean-field approach we note that $\alpha_n$ should scale as $1/N$ in order to have a finite interference in the network. This can be the case where in dense networks the number of users that interfere scales as $1/N$ or for example when an advanced receiver is used to cancel part of the interference.   This normalization is widely used in Mean field approach, e.g. \cite{mf1,mf2,mf3} and the references therein. Also, this normalization has been used and justified for instance in \cite{meshkati2006game} for a CDMA system and an Match Filter receiver. In this case, the $\SINR$ of user $n$ depends on the interference coming from other users, namely,
\begin{equation}\label{eq:interferences}
I_n(t)=\frac{1}{N}\sum_{i\neq n}p_i(t)h_i(t).
\end{equation}
\begin{proposition}
Let $I_n(t)$ be given by~\eqref{eq:interferences}, the interference perceived by transmitter $n$. We prove that 
$
I_n(t)\to I(t), \text{ as } N\to~\infty.
$
Consequently, a user in state $i$ achieves 
$
\SINR_i(t)=\frac{p_i(t)h_i}{\sum_{j=1}^{K\cdot(Q_{\text{max}}+1)}p_j(t)h_jm_j(t)+N_0}.
$
\end{proposition}
\begin{proof}
The result can be obtained by the Interchangeability property assumed in the mean field.
\end{proof}
The problem is therefore to find the power allocation policy $\phi$ such that we
\begin{align*}
\text{minimize} \int_{0}^{\infty} \bigg(\sum_{i=1}^{K\cdot(Q_{\text{max}}+1)}p_i(t)+\lambda m_i(t)\sigma(i)-E^*\bigg)\mathrm{d}t,
\end{align*}
where $E^*$ is the optimal equilibrium cost, subject to
$
\mathrm{d}m(t)=U^\phi(m(t))m(t)\mathrm{d}t, \text{ and }
p_i(t)\leq p_{\text{max}}.
$
That is we aim at characterizing a \emph{bias optimal} policy.

Next we reformulate the problem in order to be able to characterize an optimal policy $\phi$ for the problem introduced above. We note that if we were to minimize $\sum_i^{K(Q_{\text{max}}+1)}p_i(t)$ it suffices to solve
$
\SINR_i(t)=\frac{p_i(t)h_i}{N_0+\sum_{j=1}^{K(Q_{\text{max}}+1)}p_j(t)m_j(t)h_j}=\theta s_i(t), 
$
 with  $s_i(t)\in[0,1]$.
The latter has a unique solution $\vec p^*(\vec s(t))=(p_1^*(\vec s(t)),\ldots,p^*_{K(Q_{\text{max}}+1)}(\vec s(t)))$ given by
$$
p_i^*(t)=p_i^*(\vec s(t))=\frac{\theta N_0 s_i(t)}{h_i(1-\sum_{j=1}^{K(Q_{\text{max}}+1)}s_j(t)m_j(t)\theta)}, 
$$
 for all  $h_i>0,$
$p_i^*(t)=s_i(t)=0$ if $h_i=0$. For the latter solution $\vec p^*(t)$ to be a feasible solution we impose the following assumptions.
\begin{assumption}\label{ass1} We assume $\theta<1$. The latter implies
 $1-\sum_{i=1}^{K(Q_{\text{max}}+1)}s_i(t)m_i(t)\theta>0$ and $p_i^*(t)\geq 0$ for all $s_i(t)\in[0,1]$ and all $m_i(t)\in[0,1]$.
\end{assumption}
\begin{assumption}\label{ass2}
We assume $\theta\leq P_{\text{max}}h_i/(N_0+p_{\text{max}}h_i)$ for all $h_i>0$. The latter implies $p_i^*(t)\leq p_{\text{max}}$.
\end{assumption}
Therefore, we aim at finding the control vector $\vec s(t)$ such that 
\begin{align}\label{eq_obj}
\int_0^{\infty}&(\sum_{i=1}^{K(Q_{\text{max}}+1)}\frac{\theta N_0 s_i(t)}{h_i(1-\sum_{j=1}^{K(Q_{\text{max}}+1)}\theta s_j(t)m_j(t))}\nonumber\\
&+\lambda\sigma(i)m_i(t)-E^*) \mathrm{d}t,
\end{align}
is minimized, subject to $s_i(t)\in[0,1]$,  for all $i\in\{1,\ldots,Q_{\text{max}}+1\}$ and
$
\frac{\mathrm{d}\vec m(t)}{\mathrm{d}t}=\vec m(t+1)-\vec m(t)=U(\vec s(t))\vec m(t),
$
where $U(\vec s(t))$ is defined below.

\begin{proposition}\label{prop:sol_1} Let $s_i(t)$ be the action with respect to users in state $(h_i,\sigma(i))$. If $h_i=0$ or $\sigma(i)=0$ then $s_i(t)=0$ 
\end{proposition}
\begin{proof}
The proof is straightforward.
\end{proof}

The first step to obtain a bias optimal solution is to characterize an optimal equilibrium cost $E^*$. In order to do so we first compute the conditions under which objective function~\eqref{eq_obj} is convex. Let us define $(\bar m_i, \bar s_i)$ for all $i$ such that 
\begin{align}
U(\bar s_1,\ldots,\bar s_{K(Q_{\max}+1)}) \cdot(\bar m_1,\ldots,\bar m_{K(Q_{\max}+1)}
)'=0\end{align}
\begin{assumption}\label{ass3}
	For mathematical tractability, we will assume in the remaining of this section that $H=\{0,1\}$ (i.e. GOOD/BAD state) and $Q_{\max}=1$. The assumption $Q_{\max}=1$ is meaningful in the context where the transmitters are MTC (Machine Type Communications) devices or IoT objects (e.g. sensors) that transmit some updated estimations/parameters. Once a new estimation arrives, the old one in the buffer becomes useless and is dropped. In this case, we have $Q_{\max}=1$.
\end{assumption}
 In this case $\bar s_i=0$ for all $i=1,2,3$. Therefore the cost at equilibrium, $E(\vec{\bar s})$,  equals
\begin{align}\label{equili_cost}
&E(\vec{\bar s})=\frac{\theta N_0\bar s_4}{h_i(1-\theta\bar m_4)}+\lambda\bar m_4.
\end{align}
Namely,
$
E(\vec{\bar s})=C(\vec{\bar m},0)(1-\bar s_4)+C(\vec{\bar m},1)\bar s_4,
$
with $C(\bar m,0)=\lambda(\bar m_2+\bar m_4)$ and $C(\bar m,1)=\frac{\theta N_0}{1-\theta m_4}+\lambda(\bar m_2+\bar m_4)$. This equilibrium cost can be interpreted as the cost of being passive times the fraction of time the system is passive plus the cost of being active multiplied by the fraction of time the system is active.

Throughout the paper we will denote the optimal equilibrium point by $\vec m^*$ and the optimal control by $\vec s^*$. The optimal average cost, $E^*$ is therefore
$
E^*=\frac{\theta N_0s^*_4}{1-\theta m_4^*}+ \lambda(m_2^*+m_4^*).
$

We have assumed that the channel is either in a GOOD state $h=1$ or in a BAD state, i.e., $h=0$. In this case we aim at determining $\vec s(t)=(s_1(t),\ldots,s_4(t))$ for all $t$. By Proposition~\ref{prop:sol_1} we have that $s_1(t)=s_2(t)=s_3(t)=0$ for all $t$. The objective is therefore to determine $s_4(t)$ for all $t$. 

 We denote by $\rho$ the arrival probability and by $\beta_0$ the probability of being in the BAD state and by $\beta_1$ the probability of being in the GOOD state. We have that $U(\vec\bar s)$ equals
\begin{align*}
\begin{pmatrix}
-\beta_1-\beta_0\rho & 0  & \beta_0(1-\rho) & \bar s_4\beta_0(1-\rho)\\
\beta_0\rho & -\beta_1 & \beta_0\rho & \bar s_4\beta_0\rho+(1-\bar s_4)\beta_0\\
\beta_1(1-\rho) & 0 &-\beta_1\rho-\beta_0&\bar s_4\beta_1(1-\rho)\\
\beta_1\rho& \beta_1 & \beta_1\rho & -\bar s_4\beta_1(1-\rho)-\beta_0.
\end{pmatrix}
\end{align*}
By solving $U(\vec{\bar s})\cdot (\bar m_1,\bar m_2, \bar m_3,\bar m_4)'=0$ we obtain
\begin{align}\label{equilibrium_case1}
\bar m_1&= \frac{\beta_0\beta_1(1-\rho)\bar s_4}{\rho+\beta_1\bar s_4(1-\rho)};\bar m_2&= \frac{\beta_0\rho}{\rho+\beta_1\bar s_4(1-\rho)}\nonumber\\
\bar m_3&= \frac{\beta_1^2(1-\rho)\bar s_4}{\rho+\beta_1\bar s_4(1-\rho)};\bar m_4&= \frac{\beta_1\rho}{\rho+\beta_1\bar s_4(1-\rho)}
\end{align}

\begin{proposition}\label{prop:convexity} Let $\vec{\bar s}=(0,0,0,\bar s_4)$ and $\vec{\bar m}$ given by Equation~\eqref{equilibrium_case1}. Assume 
$
N_0\leq \frac{\lambda(1-\rho)(1-\beta_1\theta)^2}{\rho\theta^2}.
$
 Then $E(\vec{\bar s})$ as given by Equation~\eqref{equili_cost} is convex.
\end{proposition}
See Appendix~\ref{convexityproof} for the proof.
\subsubsection{An average optimal control}

In the next proposition we characterize the optimal equilibrium point. The result is characterized by the following constants. 
\begin{align}\label{eq:constants}
&N_0^0=\frac{\lambda\beta_1(1-\rho)(1-\theta\beta_1)}{\rho\theta},\\
&N_0^1=\frac{\lambda\beta_1(1-\rho)\rho(\rho+\beta_1-\beta_1\rho(1+\theta))^2/(\beta_1+\rho-\beta_1\rho)^2}{\theta(2\beta_1(1-\rho)\rho(1-\theta\beta_1)+\rho^2(1-\theta\beta_1)+\beta_1^2(1-\rho)^2)}
\end{align}
\begin{proposition}\label{prop:average_optimal} Let $\vec{s}^*$ be given by $(0,0,0, s_4^*)$, and let $N_0$ be as in (10) 
then  
\begin{itemize}
\item $s_4^*=0$ and $\vec{m}^*=(0,\beta_0,0,\beta_1)$ if $N_0\geq N_0^0.$
\item $s_4^*=1$ and 
\begin{align*}
&m_1^*=\frac{\beta_0\beta_1(1-\rho)}{\beta_1(1-\rho)+\rho}; \quad m_2^*=\frac{\beta_0\rho}{\rho+\beta_1(1-\rho)},\\
&m_3^*=\frac{\beta_1^2(1-\rho)}{\rho+\beta(1-\rho)};\quad m_4^*=\frac{\beta_1\rho}{\rho+\beta_1(1-\rho)}),
\end{align*}
 if $N_0\leq N_0^1.$
\item And $s_4^*\in(0,1)$ if $N_0^1<N_0<N_0^0$, with $N_0$ given by $\lambda\frac{(1-\rho)\bar m_4^2(1-\theta \bar m_4)^2}{\theta\rho(\theta\bar m_4(\bar m_4-2\beta_1)+\beta_1)}$.

\end{itemize}
\end{proposition}
See Appendix~\ref{average_optimal_policy} for the proof.


Proposition~\ref{prop:average_optimal} suggests that a threshold policy in $m_4$ is optimal for the problem presented in Equation~\eqref{eq_obj}. This is shown in the next section.

\subsubsection{Bias-optimal solution}
In this section we derive an optimal solution for the deterministic control problem in Equation~\eqref{eq_obj}. In the previous section we have characterized the optimal equilibrium point based on the value of $N_0$. In Proposition~\ref{prop:bias_optimal} (see the proof in Appendix E). We determine a bias-optimal control policy. Recall that, we are interested in this solution since the average optimal cost obtained in Proposition~\ref{prop:average_optimal} can be achieved by any control policy with equilibrium point $\vec{m}^*$.
\begin{proposition}\label{prop:bias_optimal}
Let $N_0^0$ and $N_0^1$ be given by Equation~\eqref{eq:constants}. An optimal solution for problem~\eqref{eq_obj} is:
\begin{itemize}
\item If $N_0\geq N_0^0$, then $s_4(t)=0$ if $m_4(t)\leq m_4^1$ and $s_4(t)=1$ otherwise.
\item If $N_0\leq N_0^1$, then $s_4(t)=0$ if $m_4(t)\leq m_4^0$ and $s_4(t)=1$ otherwise.
\item If $N_0\in(N_0^1,N_0^0)$, then $s_4(t)=0$ if $m_4(t)\leq m_4^*$ and $s_4(t)=1$ otherwise.
We note that $m_4^1=\beta_1\rho/(\rho+\beta_1(1-rho))$, $m_4^0=\beta_1$ and $m_4^*$ solution of $\mathrm{d}E^*(\bar s_4)/\mathrm{d}\bar s_4$.
\end{itemize}
\end{proposition}

Proposition~\ref{prop:bias_optimal} tells us that the optimal solution of the mean-field approximation is of threshold type. That is, it suffices to compare the fraction of users that have one packet to transmit and are  in a good channel state with respect to $N_0$. This is a very simple heuristic for the original stochastic problem and as we will see in the next section is nearly optimal.

\section{Numerical results}\label{section:num}
In this section we numerically evaluate our mean-field solution as proposed in Proposition~\ref{prop:bias_optimal}. We compare its performance with respect to the numerical optimal solution (obtained through Value Iteration (VI)) of the original stochastic problem as presented in the beginning of  Section~\ref{s3}. 

We consider the following example. There are 10 transmitters in the system, two possible channel qualities GOOD/BAD  and each user has at most 1 packet to transmit. The latter is motivated by MTC where each machine (e.g., sensors) has few packets (e.g., temperature) to transmit. We assume $\theta = 0.2, \beta_1 = 0.4, \lambda = 1.5$ and $N_0 = 1$. The packet arrival probability $\rho$ will vary between 0.05 and 0.3. These values satisfy Assumptions~\ref{ass1} and~\ref{ass2}. We observe in the table below that our proposed solution is nearly-optimal across all values of $\rho$. We compute the relative error $|g^{MF}-g^{VI}|*100/g^{MF}$ and the absolute error $|g^{MF}-g^{VI}|*100$, where $g^{MF}$ is the average cost incurred by our policy and $g^{VI}$ the optimal average cost computed using VI.


\noindent

\vspace{0.2cm}
\begin{tabular}{ |p{1.8cm}|p{1cm}|p{1cm}|p{1cm}| p{1cm}|}
 \hline
$\qquad\,\,\,\rho$ &$\,\,$ 0.05 & $\,\,\,\,$0.1  & $\,\,\,\,$0.2  & $\,\,\,\,$0.3\\
 \hline
Rel.Err (\%) & 0.1233 & 1.1158  & 0.0164  & 0.3494\\
Abs.Err (\%)& 0.0002 & 0.0032 & 0.00007 & 0.0019 \\\hline
\end{tabular}

\section{Conclusion}\label{section:conclusion}
We have studied the problem of power allocation in large wireless networks  taking into account both channel state information
and queue state information.  We identified an MDP formulation of this problem
and in view of the state space explosion, we performed a mean-field approximation
and let the number of devices grow to infinity so as to obtain a
deterministic control problem. By solving the HJB Equation, we derived a well-performing power allocation policy for the original stochastic problem, which turns out to be a threshold-based policy and can then be efficiently implemented in real-life wireless networks.

\bibliographystyle{plain}\bibliography{IEEEabrv,bibli1}
\newpage
\section{Appendix}

\subsection{Expressions of transition probabilities of the MDP}\label{prob_expressions}
Here we provide expressions of the transition probabilities of the MDP defined in Section~\ref{s3}. We have
\begin{align*}
&\mathbb{P}(Q_n^\phi(t+1)=q|Q_n^\phi(t)=q',\mathbf{h}(t),\mathbf{p}(t))\nonumber\\
&=\mathbb{P}(Q_n^\phi(t+1)=q'|Q_n^\phi(t)=q',\mathbf{h}(t),\mathbf{p}(t))\nonumber\\
&\quad+\mathbb{P}(Q_n^\phi(t+1)=q'+1|Q_n^\phi(t)=q',\mathbf{h}(t),\mathbf{p}(t))\nonumber\\
&\quad+\mathbb{P}(Q_n^\phi(t+1)=[q'-1,0]^+|Q_n^\phi(t)=q',\mathbf{h}(t),\mathbf{p}(t)),\nonumber
\end{align*}
where $[q'-1,0]^+=\max\{q'-1,0\}$. Recall that $\mathbb{P}(A_n(t)=1)=\rho$ we therefore have 
\begin{align*}
&\mathbb{P}(Q_n^\phi(t+1)=q'|Q_n^\phi(t)=q',\mathbf{h}(t),\mathbf{p}(t))\nonumber\\
&=\rho\mathbf{1}_{\{\SINR_n(t)\geq\theta\}}+(1-\rho)\mathbf{1}_{\{\SINR_n(t)<\theta\}},\nonumber\\
&\mathbb{P}(Q_n^\phi(t+1)=q'+1|Q_n^\phi(t)=q',\mathbf{h}(t),\mathbf{p}(t))\nonumber\\
&=\rho\mathbf{1}_{\{\SINR_n(t)<\theta\}}\nonumber\\
&\mathbb{P}(Q_n^\phi(t+1)=[q'-1,0]^+|Q_n^\phi(t)=q',\mathbf{h}(t),\mathbf{p}(t))\nonumber\\
&=(1-\rho)\mathbf{1}_{\{\SINR_n(t)\geq\theta\}}.\nonumber
\end{align*}
\subsection{Transition probabilities}\label{table}
The transition probabilities can be found in Table~I below.
\begin{table*}
\caption{Transition probabilities from state $s_i$ to $s_j$}\label{tansitions}
\centering
\begin{minipage}{0.85\textwidth}
\noindent
{\bf Independent and identically distributed channel model:}
\begin{align*}
\gamma_{i,j}^1(\mathbf{m})&=
\begin{cases}
\beta_{\ell}\rho, &\hbox{if }  j=\ell (Q_{\text{max}}+1)+r \text{ with } 1\leq r\leq Q_{\text{max}}+1 \text{ and } \ell\in\{0,\ldots,K-1\},\\
&\text{and } i=\ell' (Q_{\text{max}}+1)+r' \text{ with } r=r' \text{ and } \ell'\in\{0,\ldots,K-1\}, \\
\beta_{\ell}(1-\rho), &\hbox{if } j=\ell (Q_{\text{max}}+1)+r \text{ with } 1\leq r\leq Q_{\text{max}}+1 \text{ and } \ell\in\{0,\ldots,K-1\},\\
&\text{and } i=\ell' (Q_{\text{max}}+1)+r' \text{ with } r=\max\{r'-1,1\} \text{ and } \ell'\in\{0,\ldots,K-1\},\\
0, &\hbox{otherwise },
\end{cases}\nonumber\\
\medskip
\gamma_{i,j}^0(\mathbf{m})&= 
\begin{cases}
\beta_{\ell}(1-\rho), &\hbox{if } j=\ell (Q_{\text{max}}+1)+r \text{ with } 1\leq r\leq Q_{\text{max}}+1 \text{ and } \ell\in\{0,\ldots,K-1\}, \\
&\text{and } i=\ell' (Q_{\text{max}}+1)+r' \text{ with } r=r' \text{ and } \ell'\in\{0,\ldots,K-1\},\\
\beta_{\ell}\rho, & \hbox{if } j=\ell (Q_{\text{max}}+1)+r \text{ with } 1\leq r\leq Q_{\text{max}}+1 \text{ and } \ell\in\{0,\ldots,K-1\},\\
&\text{and } i=\ell' (Q_{\text{max}}+1)+r' \text{ with } r=\min\{r'+1,Q_{\text{max}}+1\} \text{ and } \ell'\in\{0,\ldots,K-1\},\\
0, &\hbox{otherwise}, 
\end{cases}
\end{align*}
\noindent
\textbf{Markov channel model:}
\begin{align*}
\gamma_{i,j}^1(\mathbf{m})&=
\begin{cases}
b_{\ell'\ell}\rho, & \hbox{if } j=\ell (Q_{\text{max}}+1)+r \text{ with } 1\leq r\leq Q_{\text{max}}+1 \text{ and } \ell\in\{0,\ldots,K-1\}, \\
&\text{and } i=\ell' (Q_{\text{max}}+1)+r' \text{ with } r=r' \text{ and } \ell'\in\{0,\ldots,K-1\},\\
b_{\ell'\ell}(1-\rho), & \hbox{if } j=\ell (Q_{\text{max}}+1)+r \text{ with } 1\leq r\leq Q_{\text{max}}+1 \text{ and } \ell\in\{0,\ldots,K-1\}, \\
&\text{and } i=\ell' (Q_{\text{max}}+1)+r' \text{ with } r=\min\{r'-1,1\} \text{ and } \ell'\in\{0,\ldots,K-1\},\\
0, &\hbox{otherwise },
\end{cases}\nonumber\\
\medskip
\gamma_{i,j}^0(\mathbf{m})&= 
\begin{cases}
b_{\ell'\ell}(1-\rho), &\hbox{if } j=\ell (Q_{\text{max}}+1)+r \text{ with } 1\leq r\leq Q_{\text{max}}+1 \text{ and } \ell\in\{0,\ldots,K-1\}, \\
&\text{and } i=\ell' (Q_{\text{max}}+1)+r' \text{ with } r=r' \text{ and } \ell'\in\{0,\ldots,K-1\},\\
b_{\ell'\ell}\rho, &\hbox{if }  j=\ell (Q_{\text{max}}+1)+r \text{ with } 1\leq r\leq Q_{\text{max}}+1 \text{ and } \ell\in\{0,\ldots,K-1\}, \\
&\text{and } i=\ell' (Q_{\text{max}}+1)+r' \text{ with } r=\min\{r'+1,Q_{\text{max}}+1\} \text{ and } \ell'\in\{0,\ldots,K-1\},\\
0, &\hbox{otherwise}, 
\end{cases}
\end{align*}
where  $\beta_i=\mathbb{P}(h_n(t)=c_i)$  (the probability that user $n$ is in channel state $c_i$ in the i.i.d. model).
\medskip
\hrule
\end{minipage}
\end{table*}

\subsection{Proof of Proposition~\ref{prop:convexity}}\label{convexityproof}
We denote $E(\bar s_4)=E(\vec{\bar s})$, since it only depends on the equilibrium control $\bar s_4$.
To prove convexity of $E(\bar s_4)$ is suffices to show that $\mathrm{d}^2E(\bar s_4)/\mathrm{d}s_4^2\geq 0$ for all $\bar s_4\in[0,1]$. 
Let us first compute $\partial \bar m_4/\partial \bar s_4$, namely,
$\frac{\partial \bar m_4}{\partial \bar s_4}=\frac{\beta_1^2(\rho-1)\rho}{(\rho+\beta_1\bar s_4(1-\rho))^2}.$
Therefore, the condition $\partial^2 \bar m_4/\partial \bar s_4^2\geq0$ simplifies to
$\frac{\partial^2 \bar m_4}{\partial \bar s_4^2}=\frac{2\beta_1^3(\rho-1)^2\rho}{(\rho+\beta_1\bar s_4(1-\rho))^3}\geq 0.$ To see that the latter is always greater or equal to 0 it suffices to recall that $0<\rho<1$, $0<\beta_1<1$ and $\bar s_4\in[0,1]$. 
Similarly,
$\frac{\partial^2 \bar m_2}{\partial \bar s_4^2}=\frac{2\beta_1^2\beta_0(\rho-1)^2\rho}{(\rho+\beta_1\bar s_4(1-\rho))^3}\geq 0.$ 
We now  compute $\frac{\partial^2E(\bar s_4)}{\partial \bar s_4^2}$, to do so, we first compute $\partial E(\bar s_4)/\partial \bar s_4$, namely,
\begin{align*}
 &\frac{\partial}{\partial \bar s_4}\left(\frac{\theta N_0\bar s_4}{1-\theta \bar m_4}+\lambda(\bar m_2+\bar m_4)\right)\\
&=\frac{\partial}{\partial \bar s_4}\left(\frac{\theta N_0\bar s_4}{1-\theta \frac{\beta_1\rho}{\rho+\beta_1\bar s_4(1-\rho)}}+\lambda\left(\frac{\rho}{\rho+\beta_1\bar s_4(1-\rho)}\right)\right)\\
& = \frac{\partial}{\partial \bar s_4}\left(\frac{\theta N_0\bar s_4(\rho+\beta_1\bar s_4(1-\rho))}{\rho+\beta_1\bar s_4(1-\rho)-\theta \beta_1\rho}+\frac{\lambda\rho}{(\rho+\beta_1\bar s_4(1-\rho))}\right),
\end{align*}
where the second inequality follows from the substitution of $\bar m_4=\beta_1\rho/(\rho+\beta_1\bar s_4(1-\rho))$ and $\bar m_2=(1-\beta_1)\rho/(\rho+\beta_1\bar s_4(1-\rho))$. We therefore have, after some algebra, $\frac{\partial E(\bar s_4)}{\partial s_4}$ equals
\begin{align*}
&\frac{\theta N_0(2\beta_1\rho\bar s_4(1-\rho)(1-\theta \beta_1)+\rho^2(1-\theta\beta_1)+\bar s_4^2\beta_1^2(1-\rho)^2}{(\rho+\beta_1\bar s_4(1-\rho)-\theta\beta_1\rho)^2)}\\
&-\frac{\lambda\rho\beta_1(1-\rho)}{(\rho+\beta_1\bar s_4(1-\rho))^2}.
\end{align*}
$\frac{\partial^2 E(\bar s_4)}{\partial s_4^2}$ can now easily be computed, we obtain
\begin{align*}
\frac{\partial^2 E(\bar s_4)}{\partial \bar s_4^2}=&-\frac{2\theta^2 N_0\rho^2\beta_1^2(1-\rho)(1-\beta_1\theta) }{(\rho+\beta_1\bar s_4(1-\rho)-\theta \beta_1\rho)^3}\\
&+\frac{2\lambda\rho\beta_1^2(1-\rho)^2}{(\rho+\beta_1\bar s_4(1-\rho))^3}\\
=& \beta_1^2\rho(\rho-1)\bigg(\frac{2\theta^2 N_0\rho(1-\beta_1\theta)}{(\rho+\beta_1\bar s_4(1-\rho)\quad-\theta \beta_1\rho)^3}\\
&\quad-\frac{2\lambda(1-\rho)}{(\rho+\beta_1\bar s_4(1-\rho))^3}\bigg).
\end{align*}
We want to show the latter to be $\geq 0$, and since $\rho<1$ it suffices to show
\begin{align*}
\frac{2\theta^2 N_0\rho(1-\beta_1\theta)}{(\rho+\beta_1\bar s_4(1-\rho)-\theta \beta_1\rho)^3}-\frac{2\lambda(1-\rho)}{(\rho+\beta_1\bar s_4(1-\rho))^3}\leq 0,
\end{align*}
which holds if and only if
\begin{align}\label{eq:N0conditionsatisfied}
N_0\leq f(\bar s_4),
\end{align}
with
\begin{align*}
f(\bar s_4)=\frac{\lambda(1-\rho)(\beta_1\bar s_4+\rho(1-\beta_1\bar s_4)-\rho\beta_1\theta)^3}{\theta^2\rho(1-\beta_1\theta)(\beta_1\bar s_4+\rho(1-\beta_1\bar s_4))^3}.
\end{align*}
We will now show that Inequality~\eqref{eq:N0conditionsatisfied} is implied by the condition on $N_0$ in the statement,i.e., Equation~(10). To do so we will show that $f(\bar s_4)$ is increasing in $\bar s_4\in [0,1]$, that is, $f'(\bar s_4)\geq 0$ for all $\bar s_4\in[0,1]$. We have
\begin{align*}
f'(\bar s_4) =& \frac{\lambda(1-\rho)}{\theta^2\rho(1-\beta_1\theta)}\frac{\partial }{\partial \bar s_4}\left(1-\frac{\theta\beta_1\rho}{\rho+\beta_1\bar s_4(1-\rho)}\right)^3\\
&\frac{3\lambda(1-\rho)}{\theta^2\rho(1-\beta_1\theta)}\left(1-\frac{\theta\beta_1\rho}{\rho+\beta_1\bar s_4(1-\rho)}\right)^2\\
&\cdot\frac{\partial}{\partial \bar s_4}\left(1-\frac{\theta\beta_1\rho}{\rho+\beta_1\bar s_4(1-\rho)}\right)\\
=&\frac{3\lambda(1-\rho)}{\theta(1-\beta_1\theta)}\left(1-\frac{\theta\beta_1\rho}{\rho+\beta_1\bar s_4(1-\rho)}\right)^2\frac{\beta_1(1-\rho)}{(\rho+\beta_1\bar s_4(1-\rho))^2}\\
\geq& 0.
\end{align*}
Inequality~\eqref{eq:N0conditionsatisfied} is therefore satisfied since the condition in the statement Equation~(10) 
ensures
\begin{align*}
N_0\leq \frac{\eta\mu(1-\rho)(1-\beta_1\theta)^2}{2\rho\theta^2}=f(0)\leq f(\bar s_4).
\end{align*}
Hence, 
\begin{align*}
\frac{\partial^2 E(\bar s_4)}{\partial \bar s_4^2}\geq0,
\end{align*}
and $E(\bar s_4)$ is convex in $\bar s_4$ for all $N_0\leq f(0)$.

\subsection{Proof of Proposition~\ref{prop:average_optimal}} \label{average_optimal_policy}

\begin{proof} We want to characterize the optimal equilibrium point. In Proposition~\ref{prop:convexity} we have proven that the equilibrium cost $E^*(\vec{\bar s})$ is convex in $\bar s_4$, we therefore distinguish between three possible cases, see Figure~\ref{fig_3cases}.
\begin{itemize}
\item Case 1: $\mathrm{d}E^*/\mathrm{d}\bar{s}_4\geq0$ for all $\bar s_4\in[0,1]$. In this case $s^*_4=0$.
\item Case 2: $\mathrm{d}E(\bar s_4)/\mathrm{d}\bar{s}_4=0$ for  $\bar{s}_4\in(0,1)$. In this case $s_4^*$ is such that $\mathrm{d}E(s_4^*)/\mathrm{d}s_4^*=0$.
\item Case 3: $\mathrm{d}E(\bar s_4)/\mathrm{d}\bar{s}_4$ for all $\bar s_4\in[0,1]$. In this case $s^*_4=~1$.
\end{itemize} 
We compute the first derivative of $E(\vec{\bar{s}})$ w.r.t. $\bar s_4$ which after some algebra reduces to
\begin{align}\label{eq:s_4equals0}
&\frac{\theta N_0(2\beta_1\rho\bar s_4(1-\rho)(1-\theta \beta_1)+\rho^2(1-\theta\beta_1)+\bar s_4^2\beta_1^2(1-\rho)^2)}{(\rho+\beta_1\bar s_4(1-\rho)-\theta\beta_1\rho)^2}\nonumber\\
&-\frac{\lambda\rho\beta_1(1-\rho)}{(\rho+\beta_1\bar s_4(1-\rho))^2}.
\end{align}
We know the latter is an increasing function, since its derivative w.r.t. $\bar s_4$ is $\geq 0$ (see Proposition~\ref{prop:convexity}). We then know that if $\mathrm{d}E(\vec{\bar s})/\mathrm{d}\bar s_4\geq 0$ for $\bar s_4=0$ then it is also for $\bar s_4\in(0,1]$. Similarly, if $\mathrm{d}E(\vec{\bar s})/\mathrm{d}\bar s_4\leq 0$ for $\bar s_4=1$ then also for $\bar s_4\in[0,1)$. Next we find the condition so that $\mathrm{d}E(\vec{\bar s})/\mathrm{d}\bar s_4\geq 0$ for $\bar s_4=0$. Substituting $\bar s_4=0$ in Equation~\eqref{eq:s_4equals0} we obtain 
\begin{align*}
\mathrm{d}E^*(\vec{\bar s})/\mathrm{d}\bar s_4\bigg|_{\bar s_4=0}\geq 0\Longleftrightarrow N_0\geq\frac{\lambda\beta_1(1-\rho)(1-\beta_1\theta)}{\rho\theta}.
\end{align*}
Equivalently, if we substitute $\bar s_4=1$ in Equation~\eqref{eq:s_4equals0} we obtain
\begin{align*}
&\mathrm{d}E^*(\vec{\bar s})/\mathrm{d}\bar s_4\bigg|_{\bar s_4=1}\leq 0\Longleftrightarrow \\
&N_0\leq\frac{\beta_1\lambda(1-\rho)\rho(\rho+\beta_1-\beta_1\rho(1+\theta))^2/(\beta_1+\rho-\beta_1\rho)^2}{\theta(2\beta_1(1-\rho)\rho(1-\theta\beta_1)+\rho^2(1-\theta\beta_1)+\beta_1^2(1-\rho)^2)}.
\end{align*}
\begin{figure}\centering
\includegraphics[scale=0.5]{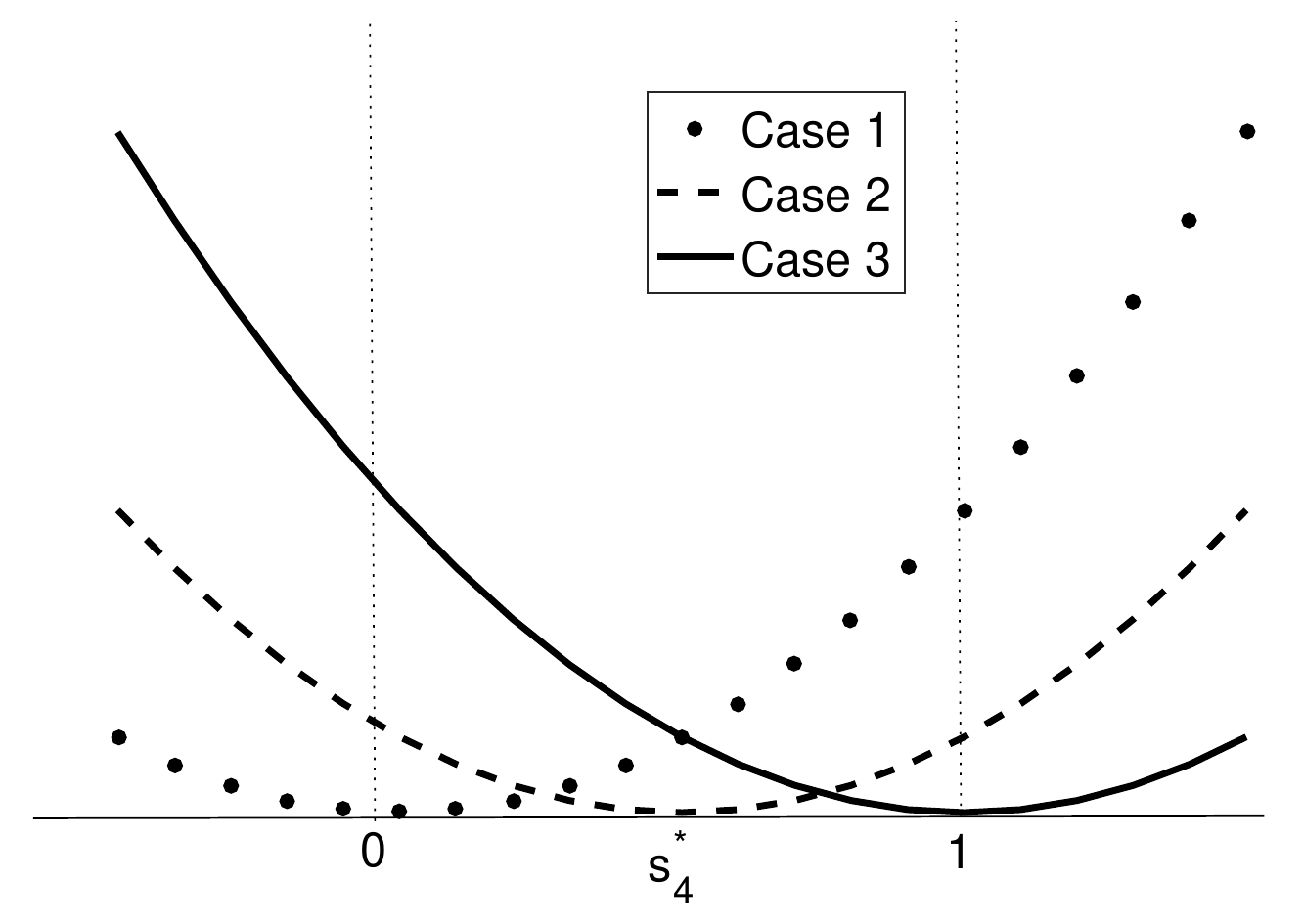}
\caption{Case~1: the cost at equilibrium is increasing in $\bar s_4$. Case~2: there exists $\bar s_4\in(0,1)$ such that $\mathrm{d}E(\bar{s}_4)/\mathrm{d}\bar{s}_4=0$. Case~3: the cost at equilibrium is decreasing in $\bar {s}_4$.}\label{fig_3cases}
\end{figure}
We have therefore proven that if $N_0\geq N_0^0$ then $s_4^*=0$, if $N_0\leq N_0^1$ then $s_4^*=1$ and if $N_0\in(N_0^1,N_0^0)$ then $s_4^*\in(0,1)$ and it is the solution obtained by equating Equation~\eqref{eq:s_4equals0} with 0, that is, 
\begin{align}\label{eq:Noexpression} 
&N_0=\frac{\lambda\rho\beta_1(1-\rho)(\rho+\beta_1\bar s_4(1-\rho)-\theta\beta_1\rho)^2/(\rho+\beta_1\bar s_4(1-\rho))^2}{\theta(2\beta_1\bar s_4(1-\rho)(1-\theta\beta_1)\rho+\rho^2(1-\theta\beta_1)+\bar s_4^2\beta_1^2(1-\rho)^2)}.
\end{align} 
The latter after substitution of $\bar s_4=\rho(\beta_1-\bar m_4)/(\beta_1(1-\rho)\bar m_4)$ yields
\begin{align}\label{eq:Nocontinues}
N_0&=\frac{\frac{\lambda\beta_1\rho(1-\rho)(\rho+\frac{\rho(\beta_1-\bar m_4)}{\bar m_4}-\theta\beta_1\rho)^2}{(\rho+\frac{\rho(\beta_1-\bar m_4)}{\bar m_4})^2}}{\theta(2\frac{\rho(\beta_1-\bar m_4)(1-\theta\beta_1)\rho}{\bar m_4}+\rho^2(1-\theta\beta_1)+\frac{\rho^2(\beta_1-\bar m_4)^2}{\bar m_4^2})},
\end{align} 
which after some algebra reduces to
\begin{align}\label{FinalN0}
N_0
&=\frac{\lambda\beta_1(1-\rho)\bar m_4^2(1-\theta \bar m_4)^2}{\theta\rho((1-\beta_1\theta)\bar m_4(2\beta_1-\bar m_4)+(\beta_1-\bar m_4)^2)}\nonumber\\
&=\frac{\lambda(1-\rho)\bar m_4^2(1-\theta \bar m_4)^2}{\theta\rho(\theta\bar m_4(\bar m_4-2\beta_1)+\beta_1)}
\end{align}


\subsection{Proof of Proposition~\ref{prop:bias_optimal}}\label{control_policy}
In order to prove that the control in the statement of Proposition~\ref{prop:bias_optimal} is optimal it suffices to show that the Hamilton-Jacobi-Bellman (HJB) equation is satisfied. The HJB equation is a partially differentiable equation that serves as sufficient condition for optimality for optimal control problems, see \cite{Bertsekas05}. The HJB equation in our particular problem reduces to the following condition,
\begin{align}\label{eq:HJBcondition}
\min\{\mathcal{V}_0(\vec m), \mathcal{V}_1(\vec m)\}=0,
\end{align}
for all $\vec m\in[0,1]^4$, where
\begin{align}
&\mathcal{V}_0(\vec m)=\lambda(m_2+m_4)-E^*+\frac{\partial V(\vec m)}{\partial \vec m}\varphi^0(\vec m),\\
&\mathcal{V}_1(\vec m)=\frac{\theta N_0}{1-\theta m_4}+\lambda(m_2+m_4)-E^*+\frac{\partial V(\vec m)}{\partial \vec m}\varphi^1(\vec m),
\end{align}
and $V(\cdot)$ the Bellman value function. In the latter equation, $\varphi^a(\vec m)$ for $a\in\{0,1\}$ represents the vector of the evolutions of the states $m_i$ $i=1,\ldots,4$, under action $a$. We will denote action $a=1$ the active action, and $a=0$ the passive action. Then 
\begin{align*}
\varphi^a(\vec m)=(\varphi^a_1(\vec m),\ldots,\varphi^a_4(\vec m))',
\end{align*}
with
\begin{align}\label{eq:evolution}
&\varphi^a_1(\vec m)=-(\beta_1+\beta_0\rho)m_1+\beta_0(1-\rho)m_3+a\beta_0(1-\rho)m_4,\nonumber\\
&\varphi^a_2(\vec m)=\beta_0\rho m_1-\beta_1m_2+\beta_0\rho m_3+ (\rho a+(1-a))\beta_0m_4,\nonumber\\
&\varphi^a_3(\vec m)=\beta_1(1-\rho)m_1-(\beta_1\rho+\beta_0)m_3+a\beta_1(1-\rho)m_4,\nonumber\\
&\varphi^a_4(\vec m)=\beta_1\rho m_1+\beta_1m_2+\beta_1\rho m_3-(\beta_1a(1-\rho)+\beta_0)m_4.
\end{align}
 Condition~\eqref{eq:HJBcondition} is written for the case $s(t)\in\{0,1\}$, as these are all the possible controls we are interested on.

 We first note that if an optimal solution satisfies the HJB equation then the following conditions must hold in all switching points:
\begin{align}
&\mathcal{V}_0(\vec m)=\mathcal{V}_1(\vec m), \label{eq:HJB_thresholdcondition1}\\
&\mathcal{V}_0(\vec m)=0,\label{eq:HJB_thresholdcondition2}\\
&\frac{\partial}{\partial m_1}\left(\frac{\partial V(\vec m)}{\partial m_4}\right)=\frac{\partial}{\partial m_4}\left(\frac{\partial V(\vec m)}{\partial m_1}\right).\label{eq:HJB_thresholdcondition3}
\end{align}
Condition~\eqref{eq:HJB_thresholdcondition1} must hold in all points $\vec m$ for which being active or passive is equally attractive, namely, in all switching points. Condition~\eqref{eq:HJB_thresholdcondition2} must be satisfied by all $\vec m$ at which passive action is optimal, in particular, for all switching points. Finally, Condition~\eqref{eq:HJB_thresholdcondition3}, symmetry of the second derivatives of the value function, must hold at all points. The partial derivative of $\partial V(\vec m)/\partial m_i$ must be continuous in the decision boundary.

Before proving that all three conditions~\eqref{eq:HJB_thresholdcondition1}-~\eqref{eq:HJB_thresholdcondition3} are satisfied we are going to show that $\frac{\partial V(\vec m)}{\partial m_1}=\frac{\partial V(\vec m)}{\partial m_3}$. By definition
\begin{align*}
V(\vec m)&=\int_0^{\infty}\bigg(\frac{\theta N_0 s_4^\pi(t)}{1-\theta m_4^\pi(t)s_4^\pi(t)}+\lambda(1-m_1^\pi(t)-m_3^\pi(t)) \\
&-E^*\bigg)\mathrm{d}t,
\end{align*}
where $\pi$ is considered to be an optimal policy. Therefore we have
\begin{align*}
\frac{\partial V(\vec m)}{\partial m_i}=&\int_0^{\infty}\frac{\partial}{\partial m_i}\bigg(\frac{\theta N_0 s_4^\pi(t)}{1-\theta m_4^\pi(t)s_4^\pi(t)}\\
&+\lambda(1-m_1^\pi(t)-m_3^\pi(t))\bigg)\mathrm{d}t,
\end{align*}
for all $i$. We are going to show that
\begin{align}\label{eq:conditionsderivatives}
&\frac{\partial(m_1^\pi(t)+m_3^\pi(t))}{\partial m_1}=\frac{\partial (m_1^\pi(t)+m_3^\pi(t))}{\partial m_3}, \hbox{and}\nonumber\\
&\frac{\partial }{\partial m_1}\left(\frac{1}{1-\theta m_4^\pi(t)}\right)=\frac{\partial }{\partial m_3}\left(\frac{1}{1-\theta m_4^\pi(t)}\right).
\end{align}
The policy $\pi$ is a combination of passive and active  intervals, therefore we will compute $m_i^{\pi,a}(t)$ in a passive time interval (when $a=0$) and in an active time interval (when $a=1$) for all $i=1,3,4$. We will later prove that Equations~\eqref{eq:conditionsderivatives} are satisfied. Note that 
\begin{align}\label{eq:ode}
\frac{\mathrm{d}m_i^{\pi,a}(t)}{\mathrm{d}t}=\varphi_i^a(\vec m^{\pi,a}(t)), \hbox{ for all } i=1,3,4.
\end{align}
We do not consider $m_2^{\pi,a}(t)$, since $m_2^{\pi,a}(t)=1-m_1^{\pi,a}(t)-m_3^{\pi,a}-m_4^{\pi,a}(t)$. If we solve the ordinary differential equation system~\eqref{eq:ode} we obtain
\begin{align*}
&m_1^{\pi,0}(t)=(\beta_1m_1-\beta_0m_3)\mathrm{e}^{-t}+\beta_0(m_1+m_3)\mathrm{e}^{-\rho t},\\
&m_3^{\pi,0}(t)=(\beta_0m_3-\beta_1m_1)\mathrm{e}^{-t}+\beta_1(m_1+m_3)\mathrm{e}^{-\rho t},\\
&m_4^{\pi,0}(t)=\beta_1(1+\mathrm{e}^{-t}(m_1+m_3-1))+\mathrm{e}^{-t}m_4\\
&\qquad\quad\quad-(m_1+m_3)\beta_1\mathrm{e}^{-\rho t},
\end{align*}
for all initial points $\vec m$, and
\begin{align*}
m_1^{\pi,1}(t)=&\frac{\mathrm{e}^{-(1+\beta_0(\rho-1))t}}{1+\beta_0(\rho-1)}\bigg(-1+m_1+m_3+m_4+\beta_0(1-m_3\\
&-m_4+m_1(\rho-1)+m_3\rho+m_4\rho)\\
&+(-1+m_3+m_4+\beta_0^2(-1+\mathrm{e}^{t})(\rho-1)\\
&-2\beta_0+\beta_0 (m_3+m_4+\rho-m_3\rho-m_4\rho))\mathrm{e}^{\beta_0(\rho-1)t}\\
&-\beta_0(\rho-1)\mathrm{e}^{t+\beta_0(\rho-1)t}\bigg),
\end{align*}
\begin{align*}
&m_3^{\pi,1}(t)\\
&=\frac{\mathrm{e}^{-(1+\beta_0(\rho-1))t}}{1+\beta_0(\rho-1)}\bigg(-1+m_1+m_3+m_4\\
&\quad+\beta_0^2(1+m_1+m_3+m_4-m_1\rho-m_3\rho-m_4\rho)\\
&\quad-\beta_0(-2+2m_3+2m_4-m_4(\rho-2)-m_3\rho-m_4\rho)\\
&\quad+(-(1+m_1+m_3+m_4)-\beta_0^3(-1+\mathrm{e}^{t})(\rho-1)\\
&\quad+(3-m_4+m_1(\rho-1)+2m_1(\rho-1)+2m_3(\rho-1)\\
&\quad-2\rho+m_4\rho)\beta_0^2\\
&\quad+\beta_0(3+m_4(\rho-2)+m_3(\rho-3)+m_1(\rho-2)-\rho))\mathrm{e}^{\beta_0(\rho-1)t}\\
&\quad+(2(\rho-1)\beta_0^2-\beta_0(\rho-1))\mathrm{e}^{t+\beta_0(\rho-1)t}\bigg),
\end{align*}
\begin{align*}
&m_4^{\pi,1}(t)\\
&=\frac{\mathrm{e}^{-t-\beta_0(\rho-1)t}}{\beta_0(1+\beta_0(\rho-1))}\bigg(
-1+m_1+m_3+m_4\\
&\quad+\beta_0^2(1-m_1+m_3+m_4)-\rho(m_1+m_3+m_4)\\
&\quad+2\beta_0(-1+m_1+m3+m4-\rho(m_1+m_3+m_4))\\
&\quad+((-1+m_1+m_3+m_4)-\beta_0^2(-1+m_1+m_3)(\rho-1)\\
&\quad+\beta_0(2-m_1(\rho-2)+m_3(\rho-2)-\rho-m_4+m_4\rho))\mathrm{e}^{\beta_0(\rho-1)t}\\
&\quad+(\beta_0^2\rho+\beta_0\rho)\mathrm{e}^{t+\beta_0(\rho-1)t}\bigg).
\end{align*}
It is now easy to show that $\frac{\partial (m_1^{\pi,a}(t)+m_3^{\pi,a}(t))}{\partial m_1}=\frac{\partial (m_1^{\pi,a}(t)+m_3^{\pi,a}(t))}{\partial m_3}$, since
\begin{align*}
&\frac{\partial (m_1^{\pi,1}(t)+m_3^{\pi,1}(t))}{\partial m_i}=\frac{\mathrm{e}^{-t} (-1 + \beta_0 + \mathrm{e}^{-\beta_0 (-1 + \rho) t)}}{\beta_0},
\end{align*}
  for all $i=1,3$ and 
	\begin{align*}
&\frac{\partial (m_1^{\pi,0}(t)+m_3^{\pi,0}(t))}{\partial m_i}=\mathrm{e}^{\rho t}, \hbox{ for all } i=1,3.
\end{align*}
 Besides, $\frac{\partial}{\partial m_1}\left(\frac{1}{1-\theta m_4^{\pi,a}(t)}\right)=\frac{\partial}{\partial m_3}\left(\frac{1}{1-\theta m_4^{\pi,a}(t)}\right)$, for $a=0,1$, since $\frac{\partial m_4^{\pi,a}(t)}{\partial m_1}=\frac{\partial m_4^{\pi,a}(t)}{\partial m_3}$ for $a=0,1$. Hence, $\frac{\partial V(\vec m)}{\partial m_1}=\frac{\partial V(\vec m)}{\partial m_3}$.

Let us now show under which conditions Equations~\eqref{eq:HJB_thresholdcondition1}-~\eqref{eq:HJB_thresholdcondition3} are satisfied. We start from Equation~\eqref{eq:HJB_thresholdcondition1}, namely,
\begin{align*}
&\mathcal{V}_0(\vec m)\\
=&\lambda(1-m_1-m_3)-E^*+\frac{\partial V(\vec m)}{\partial m_1}(\varphi_1^0(\vec m)+\varphi_3^0(\vec m))\\
&+\frac{\partial V(\vec m)}{\partial m_4}\varphi_4^0(\vec m)=\frac{\theta N_0}{1-\theta m_4}+\lambda(1-m_1-m_3)-E^*\\
&+\frac{\partial V(\vec m)}{\partial m_1}(\varphi_1^1(\vec m)+\varphi_3^1(\vec m))+\frac{\partial V(\vec m)}{\partial m_4}\varphi_4^1(\vec m)\\
=&\mathcal{V}_1(\vec m).
\end{align*}
From the latter we obtain the condition
\begin{align*}
&\frac{\partial V(\vec m)}{\partial m_1}(\varphi_1^0(\vec m)+\varphi_3^0(\vec m)-\varphi_1^1(\vec m)-\varphi_3^1(\vec m))\\
&=\frac{\theta N_0}{1-\theta m_4}+\frac{\partial V(\vec m)}{\partial m_4}(\varphi^1_4(\vec m)-\varphi^0_4(\vec m)),
\end{align*}
which after substitution of the values of $\varphi_i^a(\vec m)$, given by Equation~\eqref{eq:evolution}, gives
\begin{align}\label{eq:partial_m4}
\frac{\partial V(\vec m)}{\partial m_4}=\frac{\partial V(\vec m)}{\partial m_1}\frac{1}{\beta_1}+\frac{\theta N_0}{(1-\theta m_4)(1-\rho)\beta_1m_4}.
\end{align} 

We solve for Equation~\eqref{eq:HJB_thresholdcondition2} next. Namely,
\begin{align*}
\mathcal{V}_0(\vec m)=0\Leftrightarrow &\lambda(1-m_1-m_3)-E^*\\
&+\frac{\partial V(\vec m)}{\partial m_1}(\varphi_1^0(\vec m)+\varphi_3^0(\vec m))+\frac{\partial V(\vec m)}{\partial m_4}\varphi_4^0(\vec m)=0,
\end{align*} 
which after substitution of Equation~\eqref{eq:partial_m4} and $\varphi_i^0(\vec m)$ for all $i=1,3,4$, given by Equation~\eqref{eq:evolution}, we obtain
\begin{align*}
\frac{\partial V(\vec m)}{\partial m_1}=&\frac{(E^*-\lambda(1-m_1-m_3))\beta_1}{(\beta_1-\beta_1(m_1+m_3)-m_4)}\\
&-\frac{\theta N_0(\beta_1(\rho-1)(m_1+m_3)+\beta_1-m_4)}{(1-\theta m_4)(1-\rho)(\beta_1-\beta_1(m_1+m_3)-m_4)m_4}.
\end{align*}
Therefore, substituting the value of $\frac{\partial V(\vec m)}{\partial m_1}$ obtained above in Equation~\eqref{eq:partial_m4} we get
\begin{align*}
\frac{\partial V(\vec m)}{\partial m_4}=&\frac{\theta N_0}{(1-\theta m_4)(1-\rho)\beta_1m_4}\\
&+\frac{E^*-\lambda(1-m_1-m_3)}{\beta_1-\beta_1(m_1+m_3)-m_4}\\
&-\frac{\theta N_0(\beta_1(\rho-1)(m_1+m_3)+\beta_1-m_4)}{\beta_1(1-\theta m_4)(1-\rho)(\beta_1-\beta_1(m_1+m_3)-m_4)m_4}\\
=&\frac{E^*-\lambda(1-m_1-m_3)}{\beta_1-\beta_1(m_1+m_3)-m_4}\\
&-\frac{\theta N_0\beta_1\rho(m_1+m_3)}{\beta_1m_4(1-\theta m_4)(1-\rho)(\beta_1-\beta_1(m_1+m_3)-m_4)}.
\end{align*}
We are now left with Equation~\eqref{eq:HJB_thresholdcondition3}, that is,
\begin{align*}
\frac{\partial}{\partial m_4}\left(\frac{\partial V(\vec m)}{\partial m_1}\right)=\frac{\partial}{\partial m_1}\left(\frac{\partial V(\vec m)}{\partial m_4}\right).
\end{align*}
Let us first compute $\frac{\partial}{\partial m_4}\left(\frac{\partial V(\vec m)}{\partial m_1}\right)$, namely,
\begin{align*}
&\frac{\partial}{\partial m_4}\left(\frac{\partial V(\vec m)}{\partial m_1}\right)=\frac{\beta_1(E^*+\lambda(-1+m_1+m_3))}{(\beta_1(-1+m_1+m_3)+m_4)^2}\\
&-N_0 \theta\frac{m_4^2 - 2 m_4^3 \theta }{(m_4^2 (\beta_1 (-1 + m_1 + m_3) + 
       m_4)^2 (-1 + \rho) (-1 + m_4 \theta)^2)}\\
&-N_0 \theta\frac{( \beta_1^2 (-1 + m_1 + m_3) (1 + (m_1+m_3) (-1 + \rho)) (-1 + 
          2 m_4 \theta)} {(m_4^2 (\beta_1 (-1 + m_1 + m_3) + 
       m_4)^2 (-1 + \rho) (-1 + m_4 \theta)^2)}\\
&+ N_0 \theta\frac{\beta_1 m_4 (-2 + 4 m_4 \theta + 
         (m_1+ m_3) (2 - 2 \rho + m_4 \theta(-4+3\rho) ))}{(m_4^2 (\beta_1 (-1 + m_1 + m_3) + 
       m_4)^2 (-1 + \rho) (-1 + m_4 \theta)^2)}
\end{align*}
We now compute  $\frac{\partial}{\partial m_1}\left(\frac{\partial V(\vec m)}{\partial m_4}\right)$, that is,
\begin{align*}
&\frac{\partial}{\partial m_1}\left(\frac{\partial V(\vec m)}{\partial m_4}\right)\\
&=\frac{E^*\beta_1-\lambda m_4}{(\beta_1-\beta_1(m_1+m_3)-m_4)^2}\\
&\quad-\frac{\theta N_0\rho(\beta_1-m_4)}{(1-\theta m_4)(1-\rho)m_4(\beta_1-\beta_1(m_1+m_3)-m_4)^2}.
\end{align*}
By equating $\frac{\partial}{\partial m_4}\left(\frac{\partial V(\vec m)}{\partial m_1}\right)$ and $\frac{\partial}{\partial m_1}\left(\frac{\partial V(\vec m)}{\partial m_4}\right)$ we obtain
\begin{align}\label{eq:N0_HJBexpression}
N_0=\frac{\lambda m_4^2(\rho-1)(\theta m_4-1)^2}{\theta(\beta_1(1+(m_1+m_3)(\rho-1))(2\theta m_4-1)+m_4(1-2\theta m_4+\rho(\theta m_4-1)))}.
\end{align}
Note that in equilibrium $\bar\varphi_4(\vec \bar m) = \beta_1(\rho-1)(\bar m_1+\bar m_3)+\beta_1-\bar m_4-\beta_1\bar s_4(1-\rho)\bar m_4=0$ and $\bar s_4 = \frac{\rho(\beta_1-\bar m_4)}{\beta_1(1-\rho)\bar m_4}$, therefore
\begin{align*}
\beta_1(\rho-1)(\bar m_1+\bar m_3)+\beta_1 =\bar m_4 + \rho(\beta_1-\bar m_4).
\end{align*}
Substituting the latter in the denominator of Equation~\eqref{eq:N0_HJBexpression} we obtain
\begin{align*}
N_0&=\frac{\lambda m_4^2(\rho-1)(\theta m_4-1)^2}{\theta((\bar m_4+\rho(\beta_1-\bar m_4))(2\theta m_4-1)+m_4(1-2\theta m_4+\rho(\theta m_4-1)))}\nonumber\\
N_0&=\frac{\lambda m_4^2(\rho-1)(\theta m_4-1)^2}{\theta\rho(\theta \bar m_4(2\beta_1-\bar m_4)-\beta_1)}\nonumber\\
N_0&=\frac{\lambda m_4^2(1-\rho)(\theta m_4-1)^2}{\theta\rho(\theta \bar m_4(\bar m_4-2\beta_1)+\beta_1)}.\\
\end{align*}
The latter coincides with the $N_0$ as given by Eq.~\eqref{FinalN0}.

\end{proof}
\end{document}